\documentclass[letterpaper,journal]{IEEEtran}
\usepackage{amsmath,amssymb,amsfonts}
\usepackage{mathtools}
\usepackage{bm}
\usepackage{stfloats} 
\usepackage{float}
\usepackage{adjustbox}
\usepackage{xcolor}
\usepackage{makecell}
\usepackage{booktabs} 
\usepackage{hhline}
\usepackage[table]{xcolor} 
\usepackage{diagbox}
\usepackage{multirow}
\usepackage{array}
\usepackage{float}
\usepackage{lscape}
\usepackage{graphicx}
\usepackage{nicematrix}
\usepackage{tikz}
\usetikzlibrary{matrix,fit,calc,arrows.meta,bending,positioning,arrows}
\usepackage[ruled,vlined]{algorithm2e}
\usepackage{algpseudocode}
\usepackage[most]{tcolorbox}
\usepackage[caption=false,font=normalsize,labelfont=sf,textfont=sf]{subfig}
\usepackage{textcomp}
\usepackage{stfloats}
\usepackage{url}
\usepackage{verbatim}
\usepackage{cite}
\usepackage{mathrsfs}
\usepackage{bm}
\newcommand{\ceil}[1]{\left\lceil #1 \right\rceil}

\newtheorem{theorem}{Theorem}         
\newtheorem{lemma}[theorem]{Lemma}             
\newtheorem{definition}{Definition}

\newtheorem{remark}{Remark}

\begin{document}
	
	\title{Rate-Optimal Streaming Codes  over Three-Node Relay Networks with Burst Erasures}
	
	\author{\IEEEauthorblockN{Zhipeng Li, Wenjie Ma and Zhifang Zhang}
\thanks{The work was supported in part by the National Key R\&D Program of China under Grant 2020YFA0712300, and the CAS Project for Young Scientists in Basic Research under Grant YSBR-008.}
\thanks{The authors are with KLMM, Academy of Mathematics and Systems Science, Chinese Academy of Sciences, Beijing, 100190, China, and also with School of Mathematical Sciences, University of Chinese Academy of Sciences, Beijing, 100049, China (e-mail: lizhipeng23@mails.ucas.ac.cn; mawenjie@amss.ac.cn; zfz@amss.ac.cn).}
    }
	
	\maketitle

\begin{abstract}
	This paper investigates streaming codes over three-node relay networks under burst packet erasures with a delay constraint $T$. In any sliding window of $T+1$ consecutive packets, the source-to-relay and relay-to-destination channels may introduce burst erasures of lengths at most $b_1$ and $b_2$, respectively. Singhvi et al. proposed a construction achieving the optimal code rate when $\max\{b_1,b_2\}\mid (T-b_1-b_2)$. We construct  streaming codes with the optimal rate under the condition
	$T\geq b_1+b_2+\frac{b_1b_2}{|b_1-b_2|}$, thereby enriching the family of rate-optimal streaming codes for three-node relay networks.
\end{abstract}

	
	\newcommand{\qed}{\hfill $\square$}
	\newcommand{\dbox}[1]{%
		\begin{tikzpicture}[baseline=(X.base)]
			\node[
			draw=red,
			dashed,
			line width=0.4pt,
			inner sep=2pt,
			rounded corners=2pt,
			text width=0.95\linewidth,
			align=center,
			font=\scriptsize
			] (X) {#1};
		\end{tikzpicture}%
	}
	\newenvironment{proof}[1][Proof]{
		\par                            
		\normalfont
		\topsep6pt                       
		\trivlist                        
		\item[\hskip1em\hskip\labelsep\textit{#1:}] 
		\ignorespaces                        
	}{%
		\nobreak\hfill$\square$          
		\endtrivlist                    
	}
	
	\section{Introduction}
\IEEEPARstart{T}{he} demand for real-time communication imposes stringent requirements on low latency and highly reliable data transmission in challenging channel environments. For correcting burst erasures under strict decoding delay constraints, Martinian and Sundberg \cite{1337126} introduced streaming codes for point-to-point communication. Later, these codes were augmented with the ability to handle both random and burst erasures \cite{7593328}. Numerous works have explored various aspects of streaming codes, including theoretical bounds and code constructions \cite{8621051,9047134,9578987,8917664}, among which the diagonal embedding of block codes is a common construction approach \cite{7593328}. A recent work \cite{MaISIT25} investigated the design of generator matrices for streaming codes in a convolutional code framework, and proposed binary rate-optimal streaming codes for new parameter regimes. 

While relay nodes are prevalent in multimedia streaming and vehicular networks, streaming codes in three-node relay networks \cite{8835153} have attracted significant attention. The three-node relay network consists of a \emph{source-to-relay} (SR) channel and  a \emph{relay-to-destination} (RD) channel, with a strict delay constraint $T$ from the source to the destination. Singhvi et al. \cite{9834645} considered both random and burst erasures in the SR and RD channels, establishing an upper bound on the achievable code rate. They further developed rate-optimal code constructions for specific parameter regimes, requiring field size linear in $T$. When restricting to burst erasures, i.e., assuming the SR channel (resp. RD channel) incurs a burst erasure of length at most $b_1$ (resp. $b_2$) in any sliding window of $T+1$ consecutive packets, the work in \cite{9834645} provides the rate upper bound 
\begin{equation}\label{eq-bound}
  R\leq \min \biggl\{ \frac{T-b_1}{T-b_1+b_2}, \frac{T-b_2}{T-b_2+b_1} \biggr\} \;,
\end{equation}
and rate-optimal code when  $\max\{b_1,b_2\} \mid (T - b_1 - b_2 )$. Particularly when $b_1 = b_2$, 
binary codes that asymptotically achieve the bound (\ref{eq-bound}) have been proposed in \cite{10764771}. 

\emph{Our contribution. } Along with the generator matrix design developed in \cite{MaISIT25}, we construct binary rate-optimal streaming code over the three-node relay network for $(b_1,b_2)$ burst erasures when $T-b_1-b_2\geq \frac{b_1b_2}{|b_1-b_2|}$. Compared with the results in \cite{9834645}, our work enriches the family of binary rate-optimal three-node streaming code for $(b_1,b_2)$ burst erasures when $b_1\neq b_2$.

\section{Problem Setup}
	\subsection{Notations}
	 For $m,n\in \mathbb{N}$, denote $[m,n] = \{i \in \mathbb{N} : m \leq i \leq n\}$ and $[n] = \{i \in \mathbb{N} : 1 \leq i \leq n\}$. Let $\mathbb{F}$ be a finite field. For a matrix $P \in \mathbb{F}^{k \times n}$, $P(i,j)$ denotes the $(i,j)$-entry of $P$. For $s, t, a, b \in \mathbb{N}_+$ and $S \subseteq \mathbb{Z}$, let $\{P_i : i \in S\}$ be a set of $s \times t$ matrices and $h$ be a map from $[a] \times [b]$ to $S$. Then $\mathbf{P} = (P_{h(i,j)})_{i \in [a],\, j \in [b]}$ is a block matrix with the $(i,j)$-block being $P_{h(i,j)}$. Obviously, $\mathbf{P}$ is an $sa \times tb$ matrix. For integers  $m,n\in\mathbb{N}_+$, $m\mod n$ denotes the remainder of $m$ modulo $n$, which ranges in $[n]$. 
	

	\subsection{Point-to-Point Streaming Code}\label{sec2b}
	
	Let $b,T \in \mathbb{N}_+$. A $(b,T)$ point-to-point streaming code (SC) enables the source messages $\{S[t]\}_{t=0}^{\infty}$ to be transferred through a channel with burst erasure of length at most $b$ during every $T+1$ consecutive packets such that each packet $S[t]$ can be retrieved at time no later than $t+T$. 

To be specific, we describe a $(b,T)$ SC model under the convolutional code framework \cite{MaISIT25} as follows.
At time $t\in\mathbb{N}$, let $X[t]=(S[t],P[t])\in\mathbb{F}^n$ be the encoded packet, where $S[t] \in \mathbb{F}^k$ is the message packet and $P[t]\in\mathbb{F}^{n-k}$ is the parity packet. There exist encoding matrices $G_i\in\mathbb{F}^{k\times n}$, $i\in \mathbb{Z}$, such that 
\begin{align}
			&\quad (X[0], X[1], ..., X[t])\notag\\&=(S[0], S[1], ..., S[t]))\begin{pNiceMatrix}
        G_0 & G_1 & \cdots & G_t\\
         & G_0 & \cdots & G_{t-1}\\
       \Block{2-2}{\mathbf{0}}  &  & \ddots & \vdots\\& & &G_0
    \end{pNiceMatrix}\label{eq-1}\\
			&=(S[0], S[1], ..., S[t])( G_{j-i} )_{i,j\in [t+1]} \;.\notag
\end{align}
Since the code is systematic, it can be seen that
\begin{equation*}
    G_i=\left\{
    \begin{aligned}
      &(I_k\mid P_0),& &i=0\\
      &(\mathbf{0}_{k\times k}\mid P_i), &&\text{otherwise}
    \end{aligned}
    \right.
\end{equation*}
where $I_k$ denotes the $k\times k$ identity matrix,  \(P_i\in \mathbb F^{k\times (n-k)}\) for any \( i\in \mathbb Z\), and $P_i=\mathbf{0}$ when $i\notin [0,T]$. 
As a result, it has
\begin{equation}\label{eq1}
		P[t]=\sum _{i=0}^TS[t-i]P_i\;.
\end{equation}
Moreover, the streaming code is said to have rate $\frac{k}{n}$.

The encoded packet $X[t]$ is transmitted through the channel where erasures may happen. Fortunately, even $X[t]$ is erased, the information of $S[t]$ is still recoverable from the packets received afterwards due to (\ref{eq1}). However, a delay is thus introduced. To analyze the delay of $S[t]$, it suffices to investigate the worst case, i.e., $X[t],...,X[t+b-1]$ are erased. Then the next erasure will not happen until time $t+b+T$, which means the packets $X[t+b],...,X[t+b+T-1]$ are all successfully received. Then according to (\ref{eq1}), it has the following linear system of the erased packets:
\begin{equation}\label{recovery}
		\begin{split}
			(\bar P[t+b], \bar P[t+b+1], \ldots, \bar P[t+b+T-1]) \\
			= \left( S[t], S[t+1], \ldots, S[t+b-1] \right)\mathbf{ P},
		\end{split}
	\end{equation}
where  
	\begin{equation}\label{matrixP}
		\resizebox{0.9\linewidth}{!}{$
			\mathbf{P} = 
			\begin{pNiceMatrix}
				P_{b} & P_{b+1} & \ldots & P_{T} & 0& \cdots & 0 \\
				P_{b-1} & P_{b} & \cdots & P_{T-1} & P_{T} & \cdots &0    \\
				\vdots & \vdots & \vdots & \vdots & \vdots & \ddots & \vdots  \\
				P_1 & P_2 & \cdots & P_{T-b+1} & P_{T-b+2} & \cdots & P_T
			\end{pNiceMatrix}
			$}
	\end{equation}
and $\bar P[i]$ denotes the reduced parity packet by removing from $P[i]$ all information of the unerased packets.

Suppose $S[t]=(S_1[t],...,S_k[t])\in\mathbb{F}^k$. The $k$ coordinates of $S[t]$ may have different delay. More specifically, we define the delay profile of a $(b,T)$ SC as below.
	
\begin{definition}\label{d1}
The {\it delay profile} of a $(b,T)$ SC is a sequence $(t_1, \ldots, t_k) \in \mathbb{N}^k$ such that for any $t\in\mathbb{N}$ and $i \in [k]$, the symbol $S_i[t]$ can always be recovered no later than time $t + t_i$. It obviously has $t_i\leq T$ for all $i\in [k]$.
\end{definition}

\subsection{Three-Node Burst Erasures Correcting Streaming Codes Model}
	A $(b_1,b_2,T)$ \emph{three-node  relay network} consists of a \emph{source-to-relay} (SR) channel and  a \emph{relay-to-destination} (RD) channel, 
where the SR (resp. RD) channel permits erasure burst of length at most $b_1$ (resp. $b_2$) during every $T+1$ consecutive packets. The delay from the source to the destination must not exceed $T$. It trivially holds $T\geq b_1+b_2$.

We perform a $(b_1,b_2,T)$ three-node burst erasure-correcting
streaming code (TBSC) in three phases:
\begin{itemize}
  \item {\it The SR transmission}, realized by a $(b_1,T-b_2)$ SC encoding from $\{S[t]\}_{t=0}^{\infty}$ to $\{X[t]\}_{t=0}^{\infty}$ as in (\ref{eq-1});
  \item {\it The relay map}, which maps the received packets $Y[0],Y[1],...,Y[t]$ to a packet $R[t]\in\mathbb{F}^k$, where
  \begin{equation*}
Y[t]=\begin{cases}
\perp, & {\rm if~}X[t]{\rm ~is~erased}\\
X[t], & {\rm otherwise}
\end{cases}\;,
\end{equation*} and $R[0],R[1],...,R[t]$ are the inputs to the RD encoding;
  \item {\it The RD transmission}, realized by a $(b_2,T-b_1)$ SC encoding from $\{R[t]\}_{t=0}^{\infty}$ to $\{Z[t]\}_{t=0}^{\infty}$ as in (\ref{eq-1}).
\end{itemize}

Since the SR code and RD code both has rate $\frac{k}{n}$, the resulted TBSC in the above setup also has rate $\frac{k}{n}$. When the rate matches the bound (\ref{eq-bound}), the TBSC is called rate-optimal.

	\section{Binary Rate-Optimal Streaming Codes}
	Denote $T-b_1=pb_2+q$ where $0< q\leq b_2$. We first construct $(b_1,b_2,T)$ TBSCs for $b_1 < b_2$, achieving the optimal rate $R = \frac{T-b_1}{T-b_1+b_2}$ under the constraint $\frac{T-b_2}{b_1} \geq p + 1=\ceil{\frac{T-b_1}{b_2}}$. 
	
	\subsection{Construction for SR Code}

The SR code is given by a binary $(b_1,T-b_2)$ SC with rate $R = \frac{T-b_1}{T-b_1+b_2}$. According to Section \ref{sec2b},  we describe the code by defining the $(T-b_1)\times b_2$ binary matrix $P_i$ for $i\in[0,T-b_2]$. 
\begin{itemize}
  \item[(i)] For $i=jb_1$, $j\in[p]$, set \[
	P_{j  b_1} = 
	\begin{pmatrix}
		\mathbf{0}_{(j-1) b_2} \\
		I_{b_2} \\                 
		\mathbf{0}  
	\end{pmatrix}\footnote{Hereafter, the subscript of the zero matrix $\mathbf{0}$ indicates the row size, since its column size is clear from context. We even totally omit the subscripts when the matrix size is clear from context.}.
	\]

\item[(ii)] For $i=(p+1)b_1$, set
\[
	P_{(p+1) b_1} =
	\begin{pmatrix}
		\multicolumn{2}{@{}c@{\hspace{2\arraycolsep}}}{~~~\mathbf{0}_{pb_2}} \\
		I_{q} & \mathbf{0}
	\end{pmatrix}.
	\]
Note due to the constraint $\frac{T-b_2}{b_1}\geq p+1$, it follows that $T-b_2\geq (p+1)b_1$.
\item[(iii)] For other cases, set $P_i=\mathbf{0}$.
\end{itemize}

\begin{figure*}
		\begin{equation}\label{eqSRP}
			\mathbf{P}_{\mbox{\tiny SR}} = 
      \scalebox{0.85}{$\displaystyle 
        \begin{pNiceMatrix}
          \Block[fill=red!15,rounded-corners]{}{P_{b_1}} & P_{b_1+1} & \ldots & P_{2b_1-1} &\Block[fill=red!15,rounded-corners]{}{P_{2b_1}} & P_{2b_1+1} & \ldots & P_{3b_1-1} &~\cdots~~ & \Block[fill=red!15,rounded-corners]{}{P_{(p+1)b_1}} & 0 &\cdots & 0 \\
          P_{b_1-1} & \Block[fill=red!15,rounded-corners]{}{P_{b_1}} & \ldots & P_{2b_1-2} &P_{2b_1-1} & \Block[fill=red!15,rounded-corners]{}{P_{2b_1}} & \ldots & P_{3b_1-2} &\cdots & P_{(p+1)b_1-1} & \Block[fill=red!15,rounded-corners]{}{P_{(p+1)b_1}} &\cdots & 0 \\
          \vdots & \vdots & \ddots & \vdots &\vdots & \vdots & \ddots & \vdots & \vdots & \vdots & \vdots & \ddots & \vdots \\
          P_1 & \ldots & \ldots & \Block[fill=red!15,rounded-corners]{}{P_{b_1}} &P_{b_1+1} & \ldots & \ldots & \Block[fill=red!15,rounded-corners]{}{P_{2b_1}} &\cdots & P_{pb_1+1} & P_{(p+1)b_1+2} & \ldots & \Block[fill=red!15,rounded-corners]{}{P_{(p+1)b_1}}
        \end{pNiceMatrix}
      $}
      \end{equation}
      \end{figure*}
		
	\begin{lemma}\label{SR}
		The SR code defined in (i)-(iii) satisfies the following delay profile:
\begin{IEEEeqnarray*}{rCl}
    \IEEEeqnarraymulticol{3}{l}{
        \begin{array}{c}
            \Bigl( 
            \underbrace{b_1,\ldots,b_1}_{b_2},
            \underbrace{2b_1,\ldots,2b_1}_{b_2},\,
            \ldots,\,
            \underbrace{p b_1,\ldots,p b_1}_{b_2}, \\  
            \underbrace{(p + 1) b_1,\ldots,(p + 1) b_1}_{q} 
            \Bigr).
        \end{array}
    }
    \IEEEeqnarraynumspace
\end{IEEEeqnarray*}
	\end{lemma}
	
	\begin{proof}
It suffices to analyze the delay of $S_j[i]$ for $i\in[t, t+b_1-1]$ and $j\in[T-b_1]$ from the linear system (\ref{recovery}). According to (i)-(iii), the matrix $\mathbf{P}_{\mbox{\tiny SR}}$, which is comprised of the first $(p+1)b_1b_2$ columns of $\mathbf{P}$ in \eqref{recovery}, is displayed in (\ref{eqSRP}). Note only the blocks highlighted in red contain nonzero entries, while the remaining blocks are all zero blocks. Then, due to the structure of $P_{b_1}$ defined in (i), one can see the first $b_2$ symbols of $S[i]$, $i\in[t, t+b_1-1]$, can be recovered with delay $b_1$. Similarly, the structure of $P_{2b_1}$ indicates the next $b_2$ symbols of $S[i]$, $i\in[t, t+b_1-1]$, can be recovered with delay $2b_1$, and so on. The lemma is proved.
\end{proof}
	
\subsection{Construction for RD Code}
\begin{figure*}[ht]
		\begin{equation}\label{eqRDP}
			\mathbf{P}_{\mbox{\tiny RD}} = 
      \scalebox{0.85}{$\displaystyle 
        \begin{pNiceArray}{cccc|ccccc|c|cc|c}
          \Block[fill=red!15,rounded-corners]{}{P'_{b_2}} & P'_{b_2+1}&\ldots&P'_{b_2+q-1}&\Block[fill=red!15,rounded-corners]{}{P'_{b_2+q}}&\ldots&P'_{2b_2-1}&\ldots
          &P'_{2b_2+q-1}&\Block{7-1}{\mathbf{P}'_2}&\Block{7-2}{\ldots}&&\Block{7-1}{\mathbf{P}'_p}\\
          \Block[fill=red!15,rounded-corners]{}{P'_{b_2-1}}&\Block[fill=red!15,rounded-corners]{}{P'_{b_2}}&&
          \Block{2-1}{\vdots}&\Block{2-1}{\vdots}&\Block{2-1}{\ddots} &&&\Block{2-1}{\vdots}&&&&
          \\\Block{2-1}{\vdots}&\Block{2-1}{\vdots}&\ddots&&&&&&&&&&\\
          && &\Block[fill=red!15,rounded-corners]{}{P'_{b_2}}&P'_{b_2+1}&&\Block{2-2}{\ddots}&&\Block{2-1}{\vdots}&&&&\\
          \Block{2-1}{\vdots}&\Block{2-1}{\vdots}&&\Block{2-1}{\vdots}&\Block[fill=red!15,rounded-corners]{}{P'_{b_2}}&&&
          &&&&\\
          &&&& &\ddots&&&&&&&\\
          \Block[fill=red!15,rounded-corners]{}{P'_1}& \Block[fill=red!15,rounded-corners]{}{P'_2}&\ldots&\ldots&
          \Block[fill=red!15,rounded-corners]{}{P'_{q+1}}&\ldots&\Block[fill=red!15,rounded-corners]{}{P'_{b_2}}&\ldots
          &\Block[fill=red!15,rounded-corners]{}{P'_{b_2+q}}&&&
        \end{pNiceArray}
      $}
      \end{equation}
      \end{figure*}

The RD code is given by a binary $(b_2,T-b_1)$ SC with rate $R = \frac{T-b_1}{T-b_1+b_2}$. Again,  we describe the code by defining the $(T-b_1)\times b_2$ binary matrix $P'_i$ for $i\in[0,T-b_1]$. 
\begin{itemize}
  \item[(I)] For $i=jb_2+q$, $j\in[p]$, set
  \[
	P'_{j  b_2 + q} = 
	\begin{pmatrix}
		\mathbf{0}_{( j - 1) b_2 + q }  \\
		I_{b_2} \\                
		\mathbf{0}_{(p-j) b_2 }
	\end{pmatrix}.
	\]
\item[(II)] For $i=b_2$, set
\[
	P'_{b_2} =
	\begin{pmatrix}
		I_{q} & \mathbf{0} \\
		\multicolumn{2}{@{}c@{\hspace{2\arraycolsep}}}{\mathbf{~~~~0}_{p b_2 }}
	\end{pmatrix}.
	\]
\item[(III)] For $i\in[b_2-1]$, if $q\neq b_2$, let $d_i=i\mod q$ and $e_i=q+i\mod (b_2-q)$, then set $P'_i(d_i,e_i)=1$, while the remaining entries of $P'_i$ are all zeros. If $q=b_2$,  set  $P'_i=\mathbf{0}$ for $i\in[b_2-1]$.
\item[(IV)] For remaining $i$'s, set $P'_i=\mathbf{0}$.
\end{itemize}

Then the matrix $\mathbf{P}$ in (\ref{recovery}), here denoted as $\mathbf{P}_{\mbox{\tiny RD}}$, is displayed in (\ref{eqRDP}),
where for $2\leq i\leq p$,
\begin{equation}
\mathbf{P}'_i= \scalebox{0.85}{$\displaystyle\begin{pNiceMatrix}
          \Block[fill=red!15,rounded-corners]{}{P'_{ib_2+q}} & P'_{ib_2+q+1} & \ldots & P'_{(i+1)b_2+q-1} \\
          P'_{ib_2+q-1} & \Block[fill=red!15,rounded-corners]{}{P'_{ib_2+q}} & \ldots & P'_{(i+1)b_2+q-2} \\
          \vdots & \vdots & \ddots & \vdots \\
          P'_{(i-1)b_2+q+1} & \ldots & \ldots & \Block[fill=red!15,rounded-corners]{}{P'_{ib_2+q}} 
        \end{pNiceMatrix}$}\;.
\end{equation}

Before analyzing the delay profile, we first investigate the left-most sub-matrix (denoted as $\mathbf{P}'_{0}$) of $\mathbf{P}_{\mbox{\tiny RD}}$ separated by the first vertical line, i.e.,
\begin{equation}\label{eq9}
\mathbf{P}'_{0} = \begin{pmatrix}
		P'_{b_2} & \cdots & P'_{b_2+q-1} \\
		\vdots & \vdots & \vdots \\
		P'_{1} & \cdots & P'_{q}
	\end{pmatrix}.
\end{equation}
Note the nonzero blocks in $\mathbf{P}'_{0}$ are $P'_1,...,P'_{b_2}$, each of which contains all-zero rows except for the top $q$ rows. Furthermore, the following lemma holds.
\begin{lemma}  \label{invertible}
By deleting the bottom $T-b_1-q$ rows from each block $P'_i$ in $\mathbf{P}'_{0}$, the resulting matrix becomes an invertible $b_2q\times b_2q$  binary matrix.
\end{lemma}  
	
	\begin{proof}  
We only need to prove the resulting matrix is invertible. For simplicity, we still use the notation $P'_i$ to denote the reduced block after removing the bottom $T-b_1-q$ rows from $P'_i$. First, according to (II) and (III), one can see that
$$P'_{b_2}=(I_q~\mathbf{0}),\;\;\;\;P'_i=(\mathbf{0}~\tilde{P}_i), \;1\leq i<b_2,$$
where $\tilde{P}_i$ denotes the right $b_2-q$ columns. Put the first $q$ columns of all blocks together and perform invertible column transformations, then the resulting $b_2q\times b_2q$ matrix becomes
$$\begin{pmatrix}
  I_{q^2}&\mathbf{0}\\\mathbf{0}&\tilde{\mathbf{P}}
\end{pmatrix}, \mbox{~~~where~}
\tilde{\mathbf{P}}=\begin{pmatrix}
		\tilde{P}_{b_2-q} & \cdots & \tilde{P}_{b_2-1} \\
		\vdots & \vdots & \vdots \\
		\tilde{P}_{1} & \cdots & \tilde{P}_{q}
	\end{pmatrix}.$$
Note that $\tilde{P}_i$, $1\leq i<b_2$, has only one entry being $1$, and zeros elsewhere. Moreover, from (III) one can see that $\tilde{\mathbf{P}}$ is actually a permutation matrix (i.e., each row and each column has exactly one $1$) and therefore is invertible. The lemma follows immediately.
\end{proof}
	
\begin{lemma}\label{RD}
		The RD code defined in (I)-(IV) satisfies the following  delay profile:
\begin{IEEEeqnarray*}{rCl}
	\IEEEeqnarraymulticol{3}{l}{
		\begin{array}{c}
			\Bigl(
			\underbrace{b_2,\ldots,b_2}_{q},\,
			\underbrace{b_2\!+\!q,\ldots,b_2\!+\!q}_{b_2},\, \underbrace{2b_2\!+\!q,\ldots,2b_2\!+\!q}_{b_2},\\  
			\ldots,\,
			\underbrace{p b_2+q,\ldots,p b_2+q}_{b_2} 
			\Bigr)
		\end{array}
	}\IEEEeqnarraynumspace 
\end{IEEEeqnarray*}
	\end{lemma}
	\begin{proof}
Similar to Lemma \ref{SR}, it suffices to analyze the delay of $S_j[i]$ for $i\in[t, t+b_2-1]$ and $j\in[T-b_1]$ from the matrix $\mathbf{P}_{\mbox{\tiny RD}}$ displayed in (\ref{eqRDP}). Note that when $q=b_2$, $\mathbf{P}_{\mbox{\tiny RD}}$ is of the same form as $\mathbf{P}_{\mbox{\tiny SR}}$. So we only consider the case $q\neq b_2$, and it is accomplished in three cases below according to the value of $j$.
\begin{enumerate}
  \item $1\leq j \leq q$. 
  
  The recovery of the first $q$ symbols for all $S[i]$, $i\in[t, t+b_2-1]$ , depends on the matrix $\mathbf{P}'_0$ in (\ref{eq9}). Note the nonzero rows of $\mathbf{P}'_0$ correspond exactly to these symbols. Due to Lemma \ref{invertible}, these symbols are uniquely determined at receiving the packets $P[t+b_2], P[t+b_2+1],...,P[t+b_2+q-1]$. That is, all these symbols can be recovered by time $t+b_2+q-1$. Particularly, for $i\in[t+q,t+b_2-1]$, the recovery delay of $S_j[i]$ is thus within $b_2$. For $i\in[t, t+q-1]$, the column block of $\mathbf{P}'_0$ that contains $P'_{b_2}$ in the $(i-t+1)$-th row block contains the identity $I_q$ which corresponds exactly to the symbols $S_j[i]$, $1\leq j\leq q$, and zeros elsewhere. This implies a direct recovery of $S_j[i]$ at time $t+b_2-1+(i-t+1)=i+b_2$, so the delay is $b_2$. Therefore, for all $i\in[t, t+b_2-1]$, the first $q$ symbols of $S[i]$ can be recovered with delay $b_2$.
      
  \item $q<j< b_2+q$.
  
  Consider the sub-matrix between $\mathbf{P}'_0$ and $\mathbf{P}'_2$ in the matrix $\mathbf{P}_{\mbox{\tiny RD}}$. For simplicity, we denote this sub-matrix as $\mathbf{P}'_{1}$. It can be seen that $\mathbf{P}'_{1}$ is a $b_2\times b_2$ lower triangular matrix in the block sense with $P'_{b_2+q}$'s in the diagonal line. Moreover, these $P'_{b_2+q}$'s contain the identity matrices $I_q$ which correspond exactly to the symbols $S_j[i]$ for $q<j\leq b_2+q$ and all $i\in[t, t+b_2-1]$. Although there are nonzero entries below the $P'_{b_2+q}$'s, they all correspond to the first $q$ symbols of $S[i]$ which have been recovered before time $t+b_2+q$ in case 1). Therefore, for $q<j\leq b_2+q$, $S_j[i]$ can be recovered at time $i+b_2+q$.
  
  \item $b_2+q\leq j\leq T-b_1=pb_2+q$.
  
  Partition the interval $[b_2+q, pb_2+q]$ into $\bigcup_{l=2}^p[(l\!-\!1)b_2\!+\!q,lb_2\!+\!q\!-\!1]$.
  Note the sub-matrices $\mathbf{P}'_2,...,\mathbf{P}'_p$ are all diagonal matrices in the block sense. 
  It is easy to see that for $2\leq l\leq p$, $\mathbf{P}'_l$ indicates the symbols $S_j[i]$ for $j\in[(l\!-\!1)b_2\!+\!q,lb_2\!+\!q\!-\!1]$ can be recovered at time $i+lb_2+q$.
\end{enumerate}
\end{proof}

\subsection{The relay map and the TBSC}
Next, we define a relay map which maps the received packets of the SR code to the input packets of the RD code. The relay map is designed as a decode-and-forward map depending on the delay profile of the SR code and RD code, ensuring the delay from the source to the destination is within $T$.

Suppose the input packet of the RD code at time $t$ is $R[t]=(R_1[t],...,R_k[t])\in\mathbb{F}^k$, where $k=T-b_1=pb_2+q$ in our construction. For $i\in[k]$, define
\begin{equation}\label{eq10}
R_i[t]=S_{k\!+\!1\!-\!i}[t-l_ib_1],
\end{equation}
where $l_i=\ceil{\frac{k+1-i}{b_2}}\in[p+1]$.

\begin{theorem}\label{thm4}
The relay map in (\ref{eq10}) along with the SR code and RD code respectively described in Section III-A and III-B gives a $(b_1,b_2,T)$ TBSC with the rate $\frac{T-b_1}{T-b_1+b_2}$ which is optimal according to the upper bound (\ref{eq-bound}) when $b_1<b_2$. 
\end{theorem}

\begin{proof}
  We first show the relay map in (\ref{eq10}) is well defined which means the map value of $R_i[t]$ is computable from the received packets of the SR code by time $t$. Specifically, when $1\leq i\leq q$, $k+1-i\in[pb_2+1,pb_2+q]$. By Lemma \ref{SR}, the delay of $S_{k+1-i}[t]$ is $(p+1)b_1$, so $R_i[t]$ is computable. Similarly, one can check $R_i[t]$ for each case $\!(u\!-\!1)b_2\!+\!q<i\leq ub_2\!+\!q$, where $u\in[p]$.
  
  Next, we show the TBSC satisfies the delay $T$ which means $S_i[t]$ can be recovered at the destination by time $t+T$ for all $i\in[T-b_1]$. Due to the decode-and-forward relay map, the delay of $S_i[t]$ equals the sum of the  delay by the SR code and the delay by the RD code. Note that the relay map reverses the order of the coordinates. Combining with the delay profiles displayed in Lemma \ref{SR} and Lemma \ref{RD}, we have the following
  \begin{itemize}
    \item For $1\leq i\leq q$, the delay of $S_i[t]$ is $b_2+(p+1)b_1<pb_2+q+b_1=T$ because $b_1<b_2$ as we assumed.
    
    \item For $\!(u\!-\!1)b_2\!+\!q<i\leq ub_2\!+\!q$, where $1\leq u\leq p$, the delay of $S_i[t]$ is 
    $$ub_2+q+(p-u+1)b_1<pb_2+q+b_1=T\;.$$
  \end{itemize}
  Therefore, the delay $T$ is guaranteed at all symbols. 
  
  Finally, since the SR code can correct any $b_1$ burst erasures and the RD code can correct any $b_2$ burst erasures, and both codes share the same rate $\frac{T-b_1}{T-b_1+b_2}$, the resulting TBSC can correct any $(b_1,b_2)$ burst erasures and also has rate $\frac{T-b_1}{T-b_1+b_2}$.
  \end{proof}
  
  \begin{theorem}\label{thm5}
When $b_2<b_1$, under the constraint $\frac{T-b_1}{b_2} \geq \ceil{\frac{T-b_2}{b_1}}$, one can construct a rate-optimal binary $(b_1,b_2,T)$ TBSC.
\end{theorem}
\begin{proof}
Suppose $T-b_2=p'b_1+q'$ where $0< q'\leq b_1$. Then the SR code is a $(b_1,T-b_2)$ SC with rate $\frac{T-b_2}{T-b_2+b_1}$ as the one designed in Sec III-B, satisfying the delay profile  
\begin{IEEEeqnarray*}{rCl}
	\IEEEeqnarraymulticol{3}{l}{
		\begin{array}{c}
			\Bigl(
			\underbrace{b_1,\ldots,b_1}_{q'},\,
			\underbrace{b_1\!+\!q',\ldots,b_1\!+\!q'}_{b_1},\, \underbrace{2b_1\!+\!q',\ldots,2b_1\!+\!q'}_{b_1},\\  
			\ldots,\,
			\underbrace{p' b_1+q',\ldots,p' b_1+q'}_{b_1} 
			\Bigr)\;,
		\end{array}
	}\IEEEeqnarraynumspace 
\end{IEEEeqnarray*}
and the RD code is a $(b_2,T-b_1)$ with rate $\frac{T-b_2}{T-b_2+b_1}$ SC as the one designed in Sec III-A, satisfying the delay profile  
\begin{IEEEeqnarray*}{rCl}
\IEEEeqnarraymulticol{3}{l}{
        \begin{array}{c}
            \Bigl( 
            \underbrace{b_2,\ldots,b_2}_{b_1},
            \underbrace{2b_2,\ldots,2b_2}_{b_1},\,
            \ldots,\,
            \underbrace{p' b_2,\ldots,p' b_2}_{b_1}, \\  
            \underbrace{(p' + 1) b_2,\ldots,(p' + 1) b_2}_{q'} 
            \Bigr).
        \end{array}
    }
    \IEEEeqnarraynumspace
\end{IEEEeqnarray*}
The relay map is still the decode-and-forward map with a reverse order of the coordinates. Then one can check it is the required TBSC.
\end{proof}
	
\begin{remark}
  Combining Theorem \ref{thm4} and Theorem \ref{thm5}, we have constructed binary rate-optimal $(b_1,b_2,T)$ TBSCs under the constraint $\frac{T - \max\{b_1, b_2\}}{\min\{b_1, b_2\}} \geq \ceil{\frac{T - \min\{b_1, b_2\}}{\max\{b_1, b_2\}} }$. A sufficient condition which implies the constraint is that
  $T-b_1-b_2\geq \frac{b_1b_2}{|b_1-b_2|}$. Comparing with the constraint $\max\{b_1,b_2\}\mid T-b_1-b_2$ required in \cite{9834645}, our construction allows a more flexible $T$ when $b_1\neq b_2$. Take the example below where $b_1=2, b_2=3$. Then we can construct rate-optimal TBSCs for $T=5$ or $T\geq 7$ \footnote{It is computed from the more strict constraint $\frac{T - \max\{b_1, b_2\}}{\min\{b_1, b_2\}} \geq \ceil{\frac{T - \min\{b_1, b_2\}}{\max\{b_1, b_2\}}}$.}, while \cite{9834645} provides rate-optimal TBSCs for $3\mid (T-5)$, i.e., $T\in\{5,8,11,14,...\}$. Note that for fixed $b_1,b_2$, increasing 
$T$ allows a higher code rate. Thus, it is meaningful to consider TBSCs with a larger delay.
\end{remark}

	\section{Example}
	Set $\mathbf{b}_1 = 2$, $\mathbf{b}_2 = 3$, and $\mathbf{T} = 7$.  
	Then the optimal code rate is
	$R = \frac{T - b_1}{T - b_1 + b_2} = \frac{5}{8}$.  
	
    According to Section III-A, it has
	\[
	\mathbf{P}_{\mbox{\tiny SR}} = \begin{pmatrix}
		P_2 & P_3 & P_4 & \mathbf{0} \\
		P_1 & P_2 & P_3 & P_4 
	\end{pmatrix}
	\] 
	where
	\[
	P_2 = \begin{pmatrix}
		I_3 \\
		\mathbf{0}_{2}
	\end{pmatrix}, \quad
	P_4=\begin{pmatrix}
		\multicolumn{2}{c}{\mathbf{0}_{3}} \\
		I_2 & \mathbf{0}_{2}
	\end{pmatrix}\;,
	\]
	and $P_1=P_3=\mathbf{0}$. Then the SR code satisfies the delay profile  $(2,2,2,4,4)$.
	
	According to Section III-B, it has
	\[
	\mathbf{P}_{\mbox{\tiny RD}} = \begin{pmatrix}
		P'_3 & P'_4 & P'_5 & \mathbf{0} & \mathbf{0} \\
		P'_2 & P'_3 & P'_4 & P'_5 & \mathbf{0}  \\
		P'_1 & P'_2 & P'_3 & P'_4 & P'_5 
	\end{pmatrix}
	\]
	where 
	\[
	P'_1=\begin{pmatrix} 
       0&0&1\\
		\multicolumn{3}{c}{\mathbf{0}_{4}}
	\end{pmatrix},
	P'_2=\begin{pmatrix} 
       0&0&0\\0&0&1\\
		\multicolumn{3}{c}{\mathbf{0}_{3}}
	\end{pmatrix},
	P'_3 =
	\begin{pmatrix}  
		I_2 & \mathbf{0}_{2}\\
		\multicolumn{2}{c}{\mathbf{0}_{3}}
	\end{pmatrix},\]	
$P'_5 = \begin{pmatrix}
		\mathbf{ 0}_{2} \\
		I_3 
	\end{pmatrix}$ and $P'_4=\mathbf{0}$. The
	RD code satisfies the  delay profile $(3,3,5,5,5)$. Using a decode-and-forward relay with reversed coordinate ordering, all coordinate experience a delay of $7$.
	
	\bibliographystyle{IEEEtran}
	\bibliography{references}
\end{document}